\documentclass[12pt]{article}
\usepackage{amsmath}
\usepackage{graphicx}
\usepackage{enumerate}
\usepackage{natbib}
\usepackage{url} 

\usepackage{color}
\usepackage{verbatim}
\usepackage{amsfonts}
\usepackage{bm}
\usepackage{enumitem}
\usepackage{natbib}
\usepackage{amsthm}
\usepackage{authblk}
\usepackage{accents}
\usepackage{mathrsfs}
\usepackage{bbm}

\newcommand{\blind}{0}

\newtheorem{algorithm}{Algorithm}
\newtheorem{theorem}{Theorem}
\newtheorem{lemma}{Lemma}

\begin{document}

\def\spacingset#1{\renewcommand{\baselinestretch}%
{#1}\small\normalsize} \spacingset{1}


\if0\blind
{
  \title{\bf Data Nuggets: A Method for Reducing Big Data While Preserving Data Structure}

\author[1]{Traymon E. Beavers}
\author[2]{Ge Cheng}
\author[2]{Yajie Duan}
\author[2]{Javier Cabrera \thanks{Javier Cabrera is funded by a grant, in part, funded by the National Heart, Lung, and Blood Institute (NHLBI), National Institutes of Health (R01-HL150065).} }
\author[3]{Mariusz Lubomirski}
\author[1]{Dhammika Amaratunga}
\author[1]{Jeffrey E. Teigler}
\affil[1]{Janssen R\&D, Spring House, PA}
\affil[2]{Department of Statistics, Rutgers University, New Brunswick NJ}
\affil[3]{Amgen Pharmaceutical, Thousand Oaks, CA}

  \maketitle
} \fi

\if1\blind
{
  \bigskip
  \bigskip
  \bigskip
  \begin{center}
    {\LARGE\bf  Data Nuggets: A Method for Reducing Big Data While Preserving Data Structure}
\end{center}
  \medskip
} \fi

\bigskip
\begin{abstract}
Big data, with $N \times P$ dimension where N is extremely large, has created new challenges for data analysis, particularly in the realm of creating meaningful clusters of data. Clustering techniques, such as K-means or hierarchical clustering are popular methods for performing exploratory analysis on large datasets. Unfortunately, these methods are not always possible to apply to big data due to memory or time constraints generated by calculations of order $P*\frac{N(N-1)}{2}$. To circumvent this problem, typically the clustering technique is applied to a random sample drawn from the dataset; however, a weakness is that the structure of the dataset, particularly at the edges, is not necessarily maintained. We propose a new solution through the concept of ``data nuggets", which reduce a large dataset into a small collection of nuggets of data, each containing a center, weight, and scale parameter. The data nuggets are then input into algorithms that compute methods such as principal components analysis and clustering in a more computationally efficient manner. We show the consistency of the data nuggets based covariance estimator and apply the methodology of data nuggets to perform exploratory analysis of a flow cytometry dataset containing over one million observations using PCA  and K-means clustering for weighted observations. Supplementary materials for this article are available online.
\end{abstract}

\noindent%
\vfill

\newpage
\spacingset{1.5} 
\section{Introduction}
\label{sec:intro}

Datasets, with $N \times P$ dimension where the number of observations N is extremely large, are common in most areas of research and business including the pharmaceutical industry \citep{Sri2017}. Such data is computationally difficult to analyze so it is standard to use a smaller random sample, although the sample-to-sample variability may produce analyses with substantially different results. This motivates the use of representative samples that 
can preserve the structure of the original dataset.

Flury considered the idea of a representative sample of a distribution $F$ or a set of ``principal points'' which are a set of $k$ points that minimize the Euclidean distance from a random variable from $F$ to the closest point in the set \citep{Flu1990}. Tibshirani extended this idea to a set of ``principal lines" that approximate a dataset  \citep{Tib1992}, and Ye and Ho proposed the flowgrid to represent the dataset\citep{Ye}. 

Ghosh, Cabrera, et al. analyzed big data in the form of comorbidity binary variables for millions of patients \citep{Gho2016}. They grouped the data to reduce the number of observations and applied a weighted form of K-means clustering. Similarly, Har-Peled and Mazumdar discussed the idea 'coreset' that uses a small dataset for clustering to get an  (1+$\epsilon$)-approximation result \citep{har2004coresets}.

Independently, Mak and Joseph proposed the concept of ``support points'' to represent big data with a small subset of observations \citep{Mak2018}. These support points are the same as principal points but instead, they minimize the energy distance \citep{Sze2013}.

Immunological research has been transformed by the continued generation and improvement of flow cytometry methods \citep{Jah2012}. These methods use fluorochrome-conjugated antibodies to differentiate cell populations based on the surface and internal expression of delineating proteins. Datasets generated by standard flow cytometry experiments routinely include the interrogation of millions of cells with greater than 12 different parameters, often more. Classical flow cytometry analysis requires bivariate graphing for visualization, but it becomes inefficient and can lead to overlooking important cellular phenomena as more parameters are measured. Therefore, methods for observation reduction and cluster identification therefore become critical for managing these large datasets. 

The most typical method would be to apply a clustering technique to the dataset, such as K-means clustering or hierarchical clustering; however, a dataset as large as those found in flow cytometry experiments would require far too many resources, such as computational memory and time. As it will be shown later, earlier proposals such as random samples or support points may not be able to capture the structure around the edges of the dataset. 

We propose a different method which instead reduces the millions of data points into a smaller collection of ``data nuggets". All the individual data points coalesce into many data nuggets, while still retaining the structure of the data. A weighted form of K-means clustering can then be used to configure the data nuggets into various clusters.

Section 2 introduces notation and provides a brief overview of the issues that may arise when attempting to apply common clustering methods to large datasets. Section 3 describes the algorithm for creating data nuggets and the algorithm for creating clusters using K-means clustering for weighted observations. This section also provides simulation results comparing the accuracy of K-means clustering to K-means clustering for data nuggets with weights generated from a dataset with binary variables, compares data nuggets to the support points given by Mak and Joseph, and gives a demonstration of how well data nuggets perform on datasets when the $P$ is large using a simulated dataset. Section 4 applies the algorithm to a flow cytometry dataset containing over one million B cells. Section 5 describes two R packages created to use the method and future work that can be done concerning this method.

\section{Limitations for Large Datasets}

We now introduce notation by generalizing our motivating example of a flow cytometry dataset. Suppose an experiment with $N$ observations (B cells), where $N$ is in the millions, is conducted to measure the level of expression of $P$ different proteins. Let $\mathbf{X}$ be the matrix containing the information pertaining to the levels of expression of each protein for each B cell so that:

\[
\mathbf{X} 
=
\begin{bmatrix}
{\mathbf{\underaccent{\tilde}x}}_{1} \\
{\mathbf{\underaccent{\tilde}x}}_{2} \\
\vdots \\
{\mathbf{\underaccent{\tilde}x}}_{N-1} \\
{\mathbf{\underaccent{\tilde}x}}_{N}
\end{bmatrix} 
=
\begin{bmatrix}
x_{11} & x_{22} & \hdots & x_{1(P-1)} & x_{1P}   \\
x_{21} & x_{22} & \hdots & x_{2(P-1)} & x_{2P}   \\
\vdots & \vdots & \ddots & \vdots & \vdots \\
x_{(N-1)1} & x_{(N-1)2} & \hdots & x_{(N-1)(P-1)} & x_{(N-1)P}   \\
x_{N1} & x_{N2} & \hdots & x_{N(P-1)} & x_{NP}
\end{bmatrix}
\]
 
Where ${\mathbf{\underaccent{\tilde}x}}_n$ is the row vector containing the protein expression levels for the $n^{th}$ B cell and $x_{np}$ is the level of expression of protein $p$ for the $n^{th}$ B cell, for $n=1,2, ..., N$ and $p =1,2, ..., P$.

The goal of the experiment is to find out if there are any meaningful groups of cells based on the level of expression of the $P$ different proteins. We can search for these groups of cells by placing cells into different clusters and then determining if any of the proteins have a particularly weak or strong level of expression in any of the clusters. The memory usage and computation needed is of the order of $P*\frac{N(N-1)}{2}$ for typical clustering techniques, like K-means clustering, hierarchical clustering \citep{Cab2002} and PAM\citep{kaufman1990partitioning}.

These clustering methods have limitations for very large datasets. For K-means clustering, the final cluster assignments heavily depend on the initial choice of cluster centers \citep{Ayr2006}. A clear remedy for this is to choose multiple initial cluster centers, conduct K-means clustering, and choose the set of clusters that minimizes the total within the cluster sum of squares. For datasets with a large number of observations, many initial centers may need to be attempted. For the LLoyd, Forgy, and MacQueen algorithms \citep{Llo1982, For1965, Mac1967} the time cost is high in R \citep{R2019}, which may lead the user to sacrifice the number of initial cluster centers they choose to evaluate.

On the other hand, the Hartigan \& Wong algorithm \citep{Har1979} may fail to finish running for large datasets because the memory cost necessary to store the closest cluster assignment and the second closest cluster assignment for each observation is too high. This is also the case for hierarchical clustering methods, which may not even have a chance to begin because the distance matrix cannot be formed for datasets that are too large. 


A common solution to this problem has been to retrieve a random sample of the data and use a clustering algorithm on this reduced dataset. The intuition is that if the sample is sufficiently large, the data structure of the sampled data should match the data structure of the entire data. Unfortunately, this intuition does not always hold. Further, since a distance matrix is needed for hierarchical clustering, the random sample may need to be reduced to a very small amount of observations when compared to the entire dataset. 

Another possible solution is to reduce the large dataset to a set of only a few data points to represent the dataset as a whole, in the form of principal points, principal lines, or support points. To avoid the memory and time constraints of using full datasets or large random samples, the pitfalls associated with massive data reduction, and the lack of focus on the edges of the data structure which may occur when using support points, we propose using data nuggets.

\section{Data Nuggets}

For a dataset with dimension $N\times P$, the process of creating data nuggets is inspired by the idea of partitioning N observations into $M$ equally sized nuggets. Each nugget would represent a data nugget with P dimensional variables, where the centers of these nuggets would form the data nugget centers, the number of observations from the dataset that exist in these nuggets would form the data nugget weights, and the trace of the covariance matrices for the observations within each nugget divided by $P$ would form the data nugget scales. 

When both $P$ and $M$ are low this is a relatively simple feat, but when either $P$ or $M$ is large the amount of computational resources required becomes unrealistic. A more feasible option would be to use observations already within the dataset as the initial centers of data nuggets. This can be done by choosing observations in the dataset which are as equally spaced apart as possible. Then all the remaining observations are assigned to the data nugget with the nearest center according to a distance metric. This method will ensure that each observation will only be assigned to a single data nugget and each data nugget will contain at least one observation. Each data nugget is then re-centered by either finding the mean of the observations it contains or choosing a random observation to be the center. 

After the data nuggets are created, a refinement process is applied to them in order to split data nuggets with a large scale parameter and elongated shape. The schematics of the proposed algorithms for creating and refining data nuggets are presented in Figure \ref{fig:flowchart}.

\begin{figure}[h!]

\centerline{\includegraphics[width=400pt]{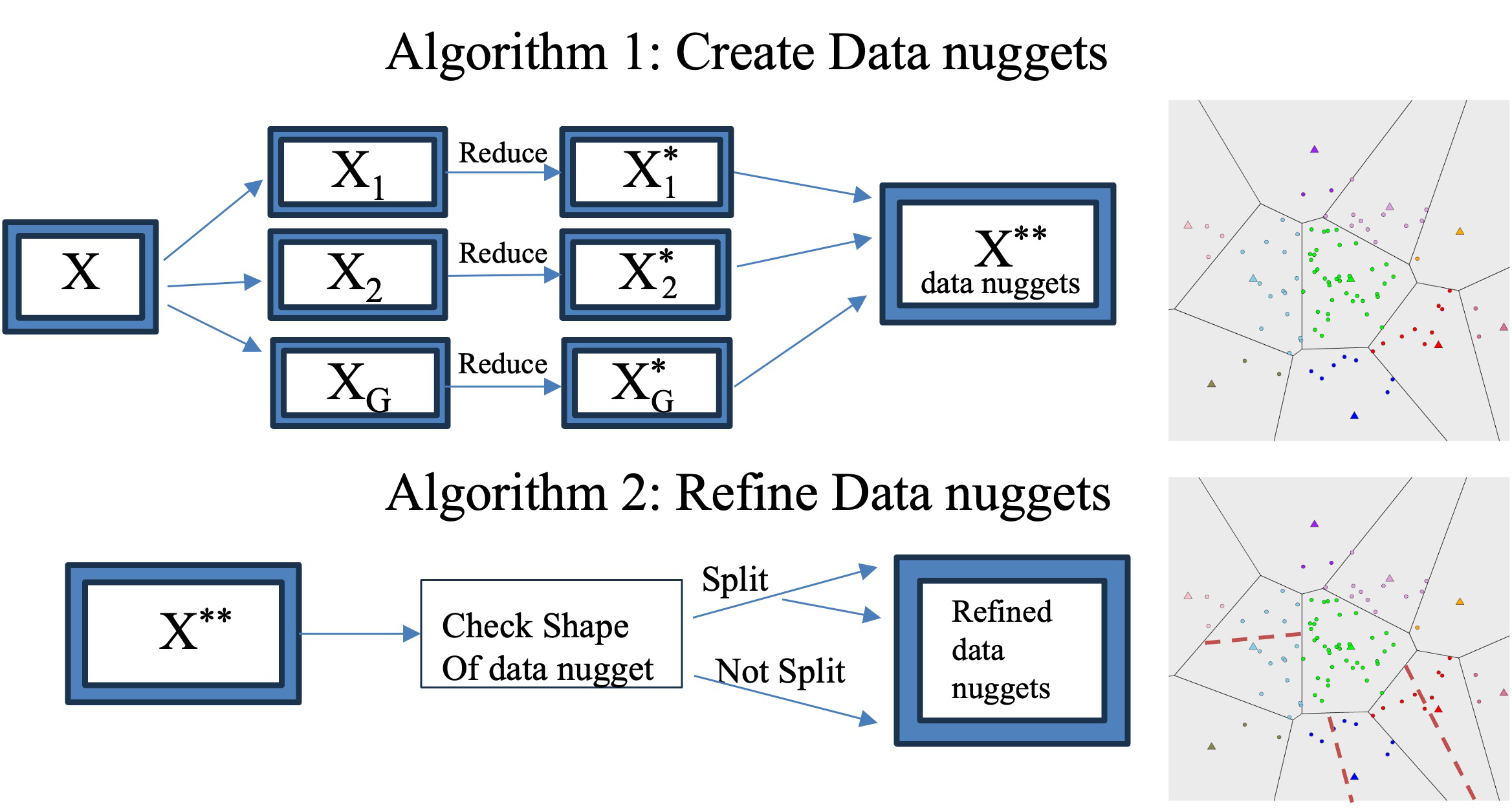}}

\caption{Algorithms to create and refine data nuggets}

\label{fig:flowchart}

\end{figure} 

Firstly, we detailed the algorithm to create data nuggets below.

\begin{algorithm} 

\label{alg1}

Create $M$ data nuggets given: $\mathbf{X}$, an $N\times P$ data matrix; a choice for how a data nugget's center will be chosen (mean or random); X is divided in G randomly selected subsets of size R. At each iteration a proportion C of observations will be deleted from each subset. $M_{init}$ is the number of observations left in total in the first part; $M$, the final number of data nuggets to create; and $D$, a distance metric.

\end{algorithm}

\begin{enumerate}


\item Randomly split $\mathbf{X}$ into the set of $R\times P$ submatrices $\{ \mathbf{X}_g | g = 1,2, ..., G \}$, where $R = \lceil \frac{N}{G} \rceil$

\item For $g = 1,2, ..., G$:\\
If the number of observations contained by $\mathbf{X}_g$ is greater than $\lceil \frac{M_{init}}{G} \rceil$, conduct steps 2.1 through 2.4 to "reduce"  $\mathbf{X}_g$ 
 to a data matrix $\mathbf{X}^{*}_g$ with numbers of observations less or equal  to $\lceil \frac{M_{init}}{G} \rceil$. Otherwise, let $\mathbf{X}^{*}_g = \mathbf{X}_g$.
	
	\begin{enumerate}[label=\theenumi.\arabic*]

\item Create a distance matrix for the observations contained in $\mathbf{X}_g$ using the given distance metric $D$.
	
\item Let $N_g$ be the current number of observations contained in $\mathbf{X}_g$. Find the $ \lfloor C N_g \rfloor$ smallest non-diagonal distances in the matrix. Let $\mathcal{A}$ and $\mathcal{B}$ be the set of $ \lfloor C N_g \rfloor$ observations from the rows and columns of the distance matrix, respectively, that have these distances between them.
	
\item Remove  $\mathcal{A}$ or $ \mathcal{B}$ from $\mathbf{X}_g$.
			
\item Repeat steps 2.1 through 2.3 until $N_g  = \lceil \frac{M_{init}}{G} \rceil$ and let $\mathbf{X}^{*}_g = \mathbf{X}_g$.

	\end{enumerate}

\item Let $(\mathbf{X}^{*})^{\prime} = [ (\mathbf{X}^{*}_1)^{\prime}, (\mathbf{X}^{*}_2)^{\prime}, ..., (\mathbf{X}^{*}_G)^{\prime} ]$ so that $\mathbf{X}^{*}$ is the data matrix containing the (roughly) $M_{init}$ centers of the initial set of data nuggets. 

\item Conduct steps 2.1 through 2.4 (replacing $\mathbf{X}_g$, $\lceil \frac{M_{init}}{G} \rceil$, and $\mathbf{X}^{*}_g$ with $\mathbf{X}^{*}$, $M$, and $\mathbf{X}^{**}$, respectively) to produce data matrix $\mathbf{X}^{**}$. The $M$ observations that remain in $\mathbf{X}^{**}$ are the initial centers of the final set of data nuggets. Let data nugget $j$ have center $\mathbf{\underaccent{\tilde}c}_j$ for $j = 1, 2, ..., M$.

\item For each observation $\mathbf{\underaccent{\tilde}x}_i$ in data matrix $\mathbf{X}$, assign $\mathbf{\underaccent{\tilde}x}_i$ to data nugget $j$ so that $D\left(\mathbf{\underaccent{\tilde}x}_i,\mathbf{\underaccent{\tilde}c}_j\right)$ is minimized over $j$. Let $N_j$ be the number of observations assigned to data nugget $j$, and let $w_j = N_j$ be the weight of data nugget $j$ for $j = 1,2, ..., M$.

\item Re-center all of the data nuggets by choosing $\mathbf{\underaccent{\tilde}c}_j$ to be either the mean of all the observations assigned to data nugget $j$ or a random observation assigned to data nugget $j$, depending on the user's choice. The later assignment may be sensitive and should be applied with caution.
	
\item Finally, let $s_j = \frac{tr\left(Cov\left(\mathbf{X}_j\right)\right)}{P}$ be the scale of data nugget $j$ when $N_j > 1$, where $\mathbf{X}_j$ is the submatrix of observations from $\mathbf{X}$ which belong to data nugget $j$. When $N_j = 1, s_j=0$.

\end{enumerate}

Here, the computational cost is dominated by the $NM$ term(see Appendix). Notice that some algorithms like K-means that are standard in statistics are of order $N^2$, therefore when N is very large, they cannot be applied. By absorbing the cost of order NxM that is required for data nuggets, the computation cost for those methods becomes of order $M^2$, which is acceptable.

The choice of $M_{init}$ is related to the hardware capabilities. One issue is the computation of distances of order $M_{init}^2$ and the memory usage, both can slow down the computation of data nuggets. In our Macbook Pro M1 with 16Gb and 4 cores, we feel comfortable using $M_{init}$ = 10000, but in systems with higher capabilities, it can be increased.

The choice of M is also similar we start with M=2000 and after the refined algorithm, we end up with around M=3000 which seems to be adequate for most datasets that we worked with. We use C=5\% as default and this is a good compromise between speed and missing data nugget centers. The choice of C will become 1 when this portion of the algorithm is implemented in C++.
In addition, in the first step, the randomness when splitting X could be reduced if we choose large R (with 5000 as default), but it will increase the computation cost in the following steps. This trade-off should be taken into consideration carefully for users.  One rational approach to deciding the number of nuggets is in the order of  $\sqrt{NP}$. For example, if $N=1,000,000$ and $P=9$, the number of data nuggets $M=3000$ is reasonable. In practice, as many data nuggets as that machine can handle are preferred. 

The first step of Algorithm 1 is to partition the data set into G subsets of observations. Different runs of algorithms may produce different data nugget sets. However, the results of applying the clustering methods to different data nugget sets will be similar.

An issue that may exist here is that the data nuggets produced may be quite large around the edges of the data because the edges region would be very sparse and the nuggets will be long and thin ellipsoids pointing toward the edges, solutions to this could be using random allocation of the center in the $6^{th}$ step or create more nuggets. In the following example, we use random allocation first and also introduce a modified algorithm later.

\begin{figure}[h!]

\centerline{\includegraphics[width=400pt]{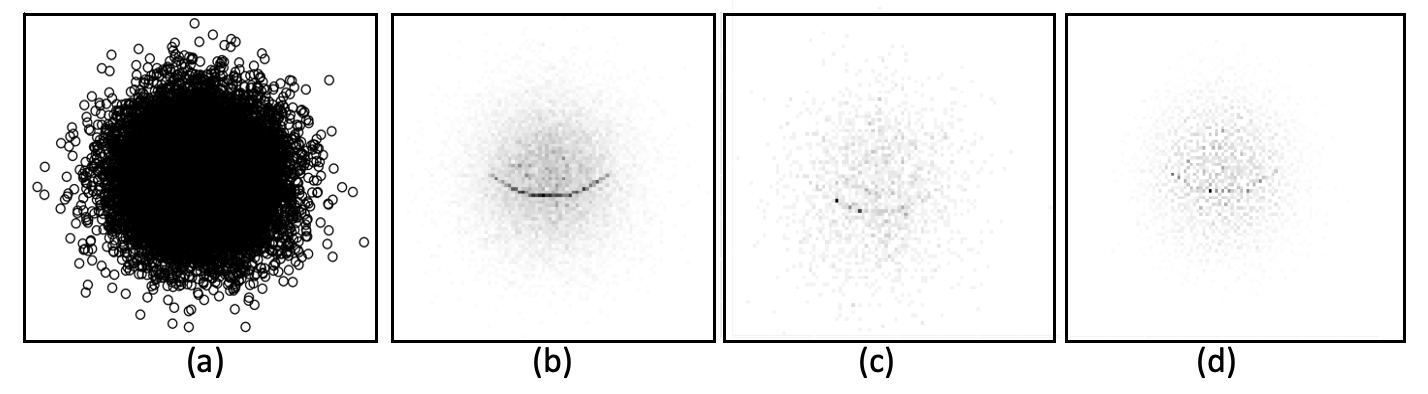}}

\caption{Comparing Density Plots for Random Sample and Data Nuggets}

\label{fig:DNExample}

\end{figure} 

Figure \ref{fig:DNExample} uses density plots to compare the amount of data structure maintained after reducing a bivariate dataset of 15,601 observations to a simple random sample of 2,000 observations versus reducing that same bivariate dataset to 2,000 data nuggets using \textbf{Algorithm \ref{alg1}}. Figure \ref{fig:DNExample} and all figures that follow were created using R. 

The data nuggets were created by reducing the entire dataset to 15,601 initial data nugget centers ($R = 5,000$, $C = 0.10$, and $M_{init}$ = 15,601) which were then reduced to 2,000 data nuggets ($M$ = 2,000) with the mean as a center and the Euclidean distance as the distance metric. The dataset is a mixture of data derived by sampling 15,000 observations from two independent standard normal distributions and combining these observations with 601 observations which create a ``smile" that is hidden inside the random noise.

The plot (a) shows a scatter plot of the entire dataset of 15,601 observations, the plot (b) shows the density plot of the entire dataset of 15,601 observations, the plot (c) shows the density plot of the random sample of 2,000 observations, and the plot (d) shows the density plot for the 2,000 data nuggets.

The density plots for the entire dataset and the random sample are created by dividing the area of the original scatterplot into a $100\times100$ grid and counting the number of observations in the dataset that fall inside each nugget of the grid. The nuggets are then color 
 ed on a gradient according to how many observations are in the nugget. Nuggets with a low number of observations produce light intensities like white or light gray while nuggets with a higher number of observations produce dark intensities such as dark gray or black.

The density plot for the data nuggets is produced in a similar manner but with a slight modification. Once again, the area of the original scatterplot is divided into a $100\times100$ grid. However, instead of using the number of data nuggets that fall within each nugget, the sum of the weights of the data nuggets that fall within the nugget is used. Then the nuggets are colored accordingly.

As shown in Figure \ref{fig:DNExample}, the density plot (b) for the entire dataset clearly shows a thin smile inside the ball of random noise. The density plot (c) for the random sample faintly produces the smile, but most of the smile is colored gray (unlike the smile in the density plot for the entire dataset which is colored black)  and there is a large amount of random noise surrounding the smile. 

The density plot (d) for the data nuggets shows a much more distinct smile. Most of this smile is colored dark gray or black and the amount of random noise is much more concentrated around the smile, matching what is seen in the density plot for the entire dataset.

Data nuggets can also be refined by splitting data nuggets with scale parameters too large and/or ``shapes" too nonspherical. The purpose of this method is to provide each data nugget with a more common level of within-data-nugget variability. The largest eigenvalue of a data nugget (i.e. the largest eigenvalue yielded from the covariance matrix of the dataset composed of the observations assigned to that data nugget) is used as a proxy for the within-data-nugget variability. Specifically, if the largest eigenvalue of a data nugget is too large, then we believe that the within-data-nugget variability is too large.

Data nuggets are refined as detailed in the algorithm below.

\begin{algorithm}
\label{alg2}

%

Refine $M$ data nuggets given: $\mathbf{X}$, an $N\times P$ data matrix; $\mathbf{X^{**}}$, the centers of $M$ data nuggets formed from $\mathbf{X}$ using \textbf{Algorithm \ref{alg1}}; $\nu$, a percentile for splitting data nuggets according to their largest eigenvalue; and $N_{min}$, the minimum number of observations that a data nugget must contain as a result of this algorithm. 

\end{algorithm}

\begin{enumerate}

\item For every data nugget:
	
	\begin{enumerate}[label=\theenumi.\arabic*]

	\item Let $\mathbf{X}_j$ be the submatrix of observations from $\mathbf{X}$ which belong to data nugget $j$.
	
	\item When $P \geq 2$ let $\zeta_j$ be the largest eigenvalue of $Cov\left(\mathbf{X}_j\right)$. When $P=1$ let $\zeta_j$ be the scale of data nugget $j$.

	\end{enumerate}

\item Obtain $\eta$, the quantile of the non-zero $\zeta_j$'s corresponding to the $\nu^{th}$ percentile. 

\item Create $\mathcal{A}$, a list of all data nuggets with scales larger than $\eta$.
	
\item For every data nugget $j\in \mathcal{A}$:
	
	\begin{enumerate}[label=\theenumi.\arabic*]

	\item If data nugget $j$ contains at least greater than $2N_{min}$ observations, split data nugget $j$ into two new data nuggets using K-means clustering.
	
	\item If either of the two new data nuggets created in step 4.1 contains less than $N_{min}$ observations, delete these two data nuggets and retain data nuggets $j$. Otherwise, delete data nugget $j$ and remove data nugget $j$ from $\mathcal{A}$.

	\end{enumerate}
	
\item Repeat steps 2 and 3 until $\mathcal{A}$ is empty or step 4 is completed without any data nuggets being removed from $\mathcal{A}$.

\end{enumerate}

Note that it is possible that data nuggets may be split an undesirable amount of times if the algorithm is left unchecked. As such, a limit can be placed on the number of times steps 2, 3, and/or 4 are executed before the algorithm ends.

Using density plots, Figure \ref{fig:RefineExample} compares the original 2,000 data nuggets refined to 2,504 data nuggets using \textbf{Algorithm \ref{alg2}} with $\nu = .5$ and $N_{min}=2$.

The first row of plots is the scatter plot of the original 2,000 data nuggets (a) besides its corresponding density plot (b). The second row of plots is the scatter plot of the refined 2,504 data nuggets (c) beside its corresponding density plot (d). The density plot for the 2,504 data nuggets has a much more consistent smile with fewer gaps and the ball of random noise is slightly more concentrated around the smile.

After the data nuggets are created, modified versions of common statistical techniques can be applied to them. In this paper, we explore K-means clustering for weighted observations and principal component analysis (PCA) for weighted observations. In both of these methods, we ignore the scale parameter of the data nuggets and instead focus on using the weight parameter. This is done because the within-data-nugget variability of the data nuggets is minuscule for a large enough number of data nuggets, provided the mean is chosen to be the data nuggets' centers. This notion is formalized in \textbf{Lemma \ref{lem1}} and \textbf{Theorem \ref{thm1}}.

\begin{figure}[h!]

\centerline{\includegraphics[width=400pt]{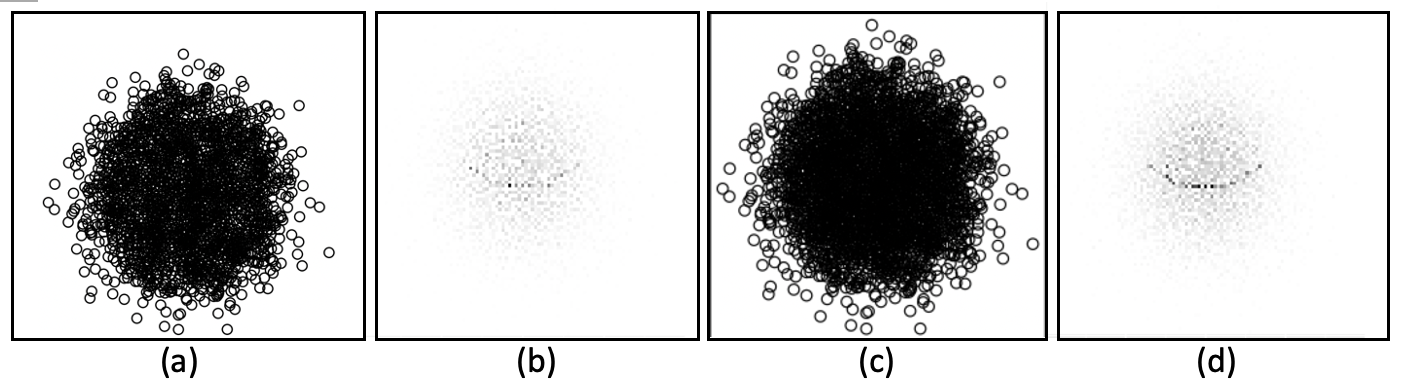}}
\caption{Comparing Density Plots for Original and Refined Data Nuggets}

\label{fig:RefineExample}

\end{figure}
Let $\mathbf{X}$ be an $N\times P$ data matrix whose rows are random vectors from some population with mean $\mathbf{\underaccent{\tilde}0}$ and covariance matrix $\mathbf{\Sigma}$. Without loss of generality, let $\mathbf{X}$ be centered at $\mathbf{\underaccent{\tilde}0}$ and let the sample covariance matrix of $\mathbf{X}$ be $\mathbf{S} = {(N-1)}^{-1}\mathbf{X}\mathbf{X}^{\prime}$. Let $\{ \mathbf{\underaccent{\tilde}c}_1, ..., \mathbf{\underaccent{\tilde}c}_M \}$, $\{w_1, ..., w_M \}$, and $\{s_1, ..., s_M \}$ be the centers, weights, and scales, respectively, of $M$ data nuggets created with \textbf{Algorithm \ref{alg1}} with data nugget centers chosen to be mean and using Euclidean distance as the choice of distance metric and refined with \textbf{Algorithm \ref{alg2}} to form a representative dataset of $\mathbf{X}$. 

Also, let these data nuggets be ordered so that $s_1 \leq s_2 \leq \cdots \leq s_M$. Further, let $\mathbf{X}_{j}$ be the submatrix of $w_j$ observations from $\mathbf{X}$ contained in data nugget $j$ for $j = 1, ..., M$. Finally, let $\mathbf{S}_{DN}= {(N-1)}^{-1}\sum_{j=1}^{M}{w_j\mathbf{c}_j\mathbf{c}_j^\prime}$, $\mathbf{\Delta}_{j} = Cov(\mathbf{X}_{j})$ for $j = 1, ..., M$, and $\mathbf{\Delta} = {(N-1)}^{-1}\sum_{j=1}^{M}(w_j - 1)\mathbf{\Delta}_j$. Note that $\mathbf{S}_{DN}$ and $\mathbf{\Delta}$ are estimates of the variability between the data nugget centers and the sum of the within-data-nugget variability of the $M$ data nuggets, respectively. Then, the following can be shown:

\begin{lemma}
\label{lem1}

If all $M$ data nuggets locally have approximately a continuous uniform distribution inside a region that is approximately a $P$-dimensional sphere and $s_1 \approx s_2 \approx \cdots \approx s_M \approx s$ for some $s > 0$ (consistent with Algorithm 2 - a more common level of within-data-nugget variability), then there is a method for splitting these $M$ data nuggets so that when $N$ increases to $2^P N$, $M$ increases to $2^P M$ and $\lim\limits_{N \rightarrow \infty} s_{M^*} = 0$ where $M^*$ is the number of data nuggets created and refined to form a representative dataset for $N$ observations.

\end{lemma}

\begin{proof}

Since $M$ data nuggets locally have approximately a continuous uniform distribution inside a $P$-dimensional sphere, let the radius of this sphere and covariance matrix for data nugget $j$ be $r_j$ and  $\mathbf{\Delta}_{j} \approx s_j^2 \mathbb{I}_P$ respectively, where $j = 1,2, ..., M$ and $\mathbb{I}_P$ is the $P \times P$ identity matrix, we have $s_j = br_j$ for some $0 < b < 1$ for $j = 1,2, ..., M$; Then, $r_1 \approx r_2 \approx \cdots \approx r_M \approx r^{\prime}$ as $s_1 \approx s_2 \approx \cdots \approx s_M \approx s$. 

Let $\mathbf{X}^{(i)}$, $N^{(i)}$, $M^{(i)}$, $s^{(i)}$ be the new data matrix at $i^{th}$ splitting step, with initial value $\mathbf{X}$, $N$, $M$, and $s$ respectively. 

First, suppose $(2^P - 1) N$ more observations are collected to X denoted by $\mathbf{X}^{(1)}$, an $N^{(1)} \times P$ data matrix, where $N^{(1)} = 2^P N$. To accommodate these new observations, $M$ data nuggets are split into $M^{(1)} = 2^P M$ data nuggets so that each new data nugget has approximately a continuous uniform distribution inside a $P$-dimensional sphere with radius $r_j \approx \frac{r^{\prime}}{2}$ for $j = 1,2, ..., M^{(1)}$. As such, $s_j \approx s^{(1)} < 2^{-1} r^{\prime}$ for $j = 1,2, ... , M^{(1)}$. 

Then, if $M^{(i-1)}$ data nuggets are split into $M^{(i)}$ data nuggets in the manner described above every time $(2^P - 1) N^{(i-1)}$ observations are added to data matrix $\mathbf{X}^{(i-1)}$ for $i = 2, 3, 4, ...$, then the $i^{th}$ time the data nuggets are split, $s_{M^{(i)}} \approx s^{(i)} < 2^{-i} r^{\prime}$. Finally, consider that $\lim\limits_{i \rightarrow \infty} 2^{-i} r^{\prime} = 0 \implies \lim\limits_{i \rightarrow \infty} s^{(i)} = 0 \implies \lim\limits_{i \rightarrow \infty} s_{M^{(i)}} = 0$ and $i \rightarrow \infty \iff N \rightarrow \infty$; therefore,  $\lim\limits_{N \rightarrow \infty} s_{M^*} = 0$ where $M^* = M^{(i)}$.

\end{proof}

\begin{theorem}
\label{thm1}

If all $M$ data nuggets locally have approximately a continuous uniform distribution inside a region that is approximately a $P$-dimensional sphere and $s_1 \approx s_2 \approx \cdots \approx s_M \approx s$ for some $s \geq 0$, then ${\mathbf{\Sigma}} = \lim\limits_{N \rightarrow \infty}{\mathbf{S}} \approx \lim\limits_{N \rightarrow \infty}{\mathbf{S}_{DN}}$.

\end{theorem}

\begin{proof}

Let:
 
\[ 
\mathbf{X} 
=
\begin{bmatrix}
\mathbf{\underaccent{\tilde}x}_1 \\
\mathbf{\underaccent{\tilde}x}_2 \\
\vdots \\
\mathbf{\underaccent{\tilde}x}_N
\end{bmatrix}
\]

where $\mathbf{\underaccent{\tilde}x}_i$ is the $1 \times P$ vector of responses for observation $i$. Further, let $\mathbf{\underaccent{\tilde}\delta}_{jk} = \mathbf{\underaccent{\tilde}c}_{j} - \mathbf{\underaccent{\tilde}x}_{jk}$ for $j = 1, 2, ..., M; k = 1, 2, ..., w_j$. First note that:

$$(N-1)\mathbf{S} = \mathbf{X}\mathbf{X}^{\prime} = \sum_{i=1}^{N}{\mathbf{\underaccent{\tilde}x}_i\mathbf{\underaccent{\tilde}x}_i^\prime}$$

Now observe

$$\sum_{i=1}^{N}{\mathbf{\underaccent{\tilde}x}_i\mathbf{\underaccent{\tilde}x}_i^\prime} = \sum_{k=1}^{w_1}{\mathbf{\underaccent{\tilde}x}_{1k}\mathbf{\underaccent{\tilde}x}_{1k}^\prime} + \sum_{k=1}^{w_2}{\mathbf{\underaccent{\tilde}x}_{2k}\mathbf{\underaccent{\tilde}x}_{2k}^\prime} + \cdots +  \sum_{k=1}^{w_M}{\mathbf{\underaccent{\tilde}x}_{Mk}\mathbf{\underaccent{\tilde}x}_{Mk}^\prime}$$

$$\sum_{j=1}^{M}{\sum_{k=1}^{w_j}{\mathbf{\underaccent{\tilde}x}_{jk}\mathbf{\underaccent{\tilde}x}_{jk}^\prime}} = \sum_{j=1}^{M}{\sum_{k=1}^{w_j}{(\mathbf{\underaccent{\tilde}c}_{j} - \mathbf{\underaccent{\tilde}\delta}_{jk})(\mathbf{\underaccent{\tilde}c}_{j} - \mathbf{\underaccent{\tilde}\delta}_{jk})^\prime}}$$

$$\sum_{j=1}^{M}{\sum_{k=1}^{w_j}{\big( \mathbf{\underaccent{\tilde}c}_{j}\mathbf{\underaccent{\tilde}c}_{j}^\prime - \mathbf{\underaccent{\tilde}c}_{j}\mathbf{\underaccent{\tilde}\delta}_{jk}^\prime - \mathbf{\underaccent{\tilde}\delta}_{jk}\mathbf{\underaccent{\tilde}c}_{j}^\prime + \mathbf{\underaccent{\tilde}\delta}_{jk}\mathbf{\underaccent{\tilde}\delta}_{jk}^\prime \big)}}$$

Since the data nuggets are re-centered after all observations are assigned, $\sum_{k=1}^{w_j}{\mathbf{\underaccent{\tilde}c}_{j}\mathbf{\underaccent{\tilde}\delta}_{jk}^\prime} = \sum_{k=1}^{w_j}{\mathbf{\underaccent{\tilde}\delta}_{jk}\mathbf{\underaccent{\tilde}c}_{j}^\prime} = \mathbf{0}_P$ where $\mathbf{0}_P$ is the $P \times P$ matrix of zeroes. So:

$$\sum_{j=1}^{M}{\sum_{k=1}^{w_j}{\big( \mathbf{\underaccent{\tilde}c}_{j}\mathbf{\underaccent{\tilde}c}_{j}^\prime - \mathbf{\underaccent{\tilde}c}_{j}\mathbf{\underaccent{\tilde}\delta}_{jk}^\prime - \mathbf{\underaccent{\tilde}\delta}_{jk}\mathbf{\underaccent{\tilde}c}_{j}^\prime + \mathbf{\underaccent{\tilde}\delta}_{jk}\mathbf{\underaccent{\tilde}\delta}_{jk}^\prime \big)}} = \sum_{j=1}^{M}{{\sum_{k=1}^{w_j}\mathbf{\underaccent{\tilde}c}_{j}\mathbf{\underaccent{\tilde}c}_{j}^\prime}} + \sum_{j=1}^{M}{\sum_{k=1}^{w_j}{\mathbf{\underaccent{\tilde}\delta}_{jk}\mathbf{\underaccent{\tilde}\delta}_{jk}^\prime}}$$

Observe that $\mathbf{\Delta}_j = \sum_{k=1}^{w_j}{\mathbf{\underaccent{\tilde}\delta}_{jk}\mathbf{\underaccent{\tilde}\delta}_{jk}^\prime}$ for $j=1, 2, ..., M$. So:

$$\sum_{j=1}^{M}{{\sum_{k=1}^{w_j}\mathbf{\underaccent{\tilde}c}_{j}\mathbf{\underaccent{\tilde}c}_{j}^\prime}} + \sum_{j=1}^{M}{\sum_{k=1}^{w_j}{\mathbf{\underaccent{\tilde}\delta}_{jk}\mathbf{\underaccent{\tilde}\delta}_{jk}^\prime}} = \sum_{j=1}^{M}{{w_j\mathbf{\underaccent{\tilde}c}_{j}\mathbf{\underaccent{\tilde}c}_{j}^\prime}} +  \sum_{j=1}^{M}\mathbf{\Delta}_j$$

and

$$\mathbf{S} = (N-1)^{-1}\sum_{j=1}^{M}{{w_j\mathbf{\underaccent{\tilde}c}_{j}\mathbf{\underaccent{\tilde}c}_{j}^\prime}} +  (N-1)^{-1}\sum_{j=1}^{M}\mathbf{\Delta}_j = \mathbf{S}_{DN} +  (N-1)^{-1}\sum_{j=1}^{M}\mathbf{\Delta}_j$$

Since all $M$ data nuggets have approximately identical scales equivalent to $s$ for some $s > 0$, $\sum_{j=1}^{M}\mathbf{\Delta}_j \approx Ms^2 \mathbb{I}_P$ and $\mathbf{\Delta}  \approx {(N-1)}^{-1}(N-M)s^2 \mathbb{I}_P$, so:

$$\mathbf{S} \approx \mathbf{S}_{DN} + M(N-M)^{-1}\mathbf{\Delta}$$

By \textbf{Lemma \ref{lem1}} and using the method provided in the proof of \textbf{Lemma \ref{lem1}}:

$$\lim_{N \rightarrow \infty} M(N-M)^{-1}\mathbf{\Delta} \approx  \lim_{N \rightarrow \infty} M(N-1)^{-1} s^{2} \mathbb{I}_P \approx \lim_{N \rightarrow \infty} M(N-1)^{-1} s^2_{M^*} \mathbb{I}_P =  \mathbf{0}_P$$

Therefore:

$${\mathbf{\Sigma}} = \lim_{N \rightarrow \infty} \mathbf{S} \approx \lim_{N \rightarrow \infty} \mathbf{S}_{DN} + \lim_{N \rightarrow \infty} M(N-M)^{-1} \mathbf{\Delta} = \lim_{N \rightarrow \infty} \mathbf{S}_{DN}$$
\end{proof}

By \textbf{Theorem \ref{thm1}}, given a large enough sample size and a large enough number of corresponding data nuggets, the amount of variability within the dataset is preserved when the dataset is reduced without accounting for the within-data-nugget variability of each data nugget. As such, this within-data-nugget variability can be ignored when applying statistical techniques to the data nuggets.

\subsection{K-means Clustering for weighted observations}

Using the idea from K means, we now introduce a clustering algorithm - K-means Clustering for weighted observations such as Data Nuggets. It is worth noting that other weighted K-means clustering methods have been developed. An example is an algorithm that can be used for analyzing social networks \citep{Liu2014}. This algorithm is designed for the purpose of finding clusters of nodes in a social network where weights are assigned to the edges that connect the nodes. The weights of the edges are described as the ``intimacy" level between the two nodes that the edge connects. In our algorithm, the weights of each data nugget are a measure of how many observations from the original dataset are contained in the data nugget. 

We describe a method of K-means clustering for weighted observations to form clusters of data nuggets with the algorithm below.

\begin{algorithm}
\label{alg3}

Conduct K-means clustering for weighted observations to form clusters of data nuggets given: $M$ data nuggets; $K$, the number of clusters to be created; $\mathbf{\underaccent{\tilde}w}$, the $M \times 1$ vector containing the weight of each data nugget.

\end{algorithm}

Let the total weighted within-cluster sum of squares, $\Omega$, be defined by 

$$\Omega \equiv \sum_{k=1}^{K}\sum_{\mathbf{\underaccent{\tilde}x} \in L_k}{\mathbf{\underaccent{\tilde}w}^{\prime}(\mathbf{\underaccent{\tilde}x} - \mathbf{\underaccent{\tilde}\mu}_k)^{\prime}(\mathbf{\underaccent{\tilde}x} - \mathbf{\underaccent{\tilde}\mu}_k)}$$

where $\{L_k | j = 1, 2, ..., K  \}$ is the set of $K$ clusters and $\{ \mathbf{\underaccent{\tilde}\mu}_k | k = 1, 2, ..., K \}$ the set of respective cluster centers.

\begin{enumerate}

\item Choose $K$ data nugget centers to be the initial cluster centers, $\mathbf{\underaccent{\tilde}\mu}_{01}, \mathbf{\underaccent{\tilde}\mu}_{02}, ..., \mathbf{\underaccent{\tilde}\mu}_{0K}$, of $K$ clusters, $L_1,L_2, ...,L_K$, respectively. 

\item For $j = 1, 2, ...,  M$, assign data nugget $j$ to the cluster $L_k$ for which $D(\mathbf{\underaccent{\tilde}c}_j,\mathbf{\underaccent{\tilde}\mu}_{0k})$ is minimized over $k = 1, 2, ... K$ where $\mathbf{\underaccent{\tilde}c}_j$ is the center of data nugget $j$  and $D$ is the Euclidean distance metric.

\item Recalculate the cluster centers $\mathbf{\underaccent{\tilde}\mu}_1, \mathbf{\underaccent{\tilde}\mu}_2,..., \mathbf{\underaccent{\tilde}\mu}_K$ as the mean of the centers of all the data nuggets within clusters $L_1,L_2, ..., L_K$, respectively.
	
\item For $J = 1, 2,... ,  M $   data nuggets:

	\begin{enumerate}[label=\theenumi.\arabic*]
	
	\item Retrieve the cluster assignment for data nugget $j$, $L_{\left(k^{\prime}\right)}$, calculate $\Omega$, and let $\Omega_{\left(k^{\prime}\right)} = \Omega$.
	
	\item Reassign data nugget $j$ to every cluster in $\{L_1,L_2, ..., L_K\}\setminus L_{\left(k^{\prime}\right)}$ and calculate $\Omega$ for each of the $K-1$ possible reassignments, $\{\Omega_1,\Omega_2, ...,\Omega_K\}\setminus \Omega_{\left(k^{\prime}\right)}$, where $\Omega_k$ is the total weighted within-cluster sum of squares when data nugget $j$ is assigned to cluster $L_k$ for $k = 1,2, ..., K$ and $k\neq\left(k^{\prime}\right)$.
	
	\item If $\Omega_k=min\left(\{\Omega_1,\Omega_2,... ,\Omega_K\}\right)<\Omega_{\left(k^{\prime}\right)}$, reassign data nugget $j$ to cluster $L_k$ and recalculate the cluster centers $\mathbf{\underaccent{\tilde}\mu}_1, \mathbf{\underaccent{\tilde}\mu}_2, ..., \mathbf{\underaccent{\tilde}\mu}_K$ as the mean of the centers of all the data nuggets within clusters $L_1,L_2, ...,L_K$, respectively.

	\end{enumerate}
	
\item Repeat step 4 until step 4 is completed without executing step 4.3.

\end{enumerate}

The outcome of \textbf{Algorithm \ref{alg3}} can be improved by repeating the algorithm with multiple choices for the initial centers chosen in step 1. The clustering assignments that minimize the total weighted within the cluster sum of squares would then be chosen as the clustering configuration. Further, it may take an extremely long time for the algorithm to converge. As such, a limit can be placed on the number of times step 4 is executed before the algorithm ends. To illustrate the usefulness of \textbf{Algorithm \ref{alg3}}, we conducted numerical analyses using example 1 and example 2 in the following part.  
\newline
\textbf{Example 1}\\ 
This simulated dataset is meant to mimic a list of 300,000 patients and whether they suffer from a list of ten conditions, and it is used to compare the performances of weighted k-means with data nuggets and k-means with raw data. We separate the data into three clusters based on the set of conditions these patients suffer from. Let $p$ be the probability of having any condition and three clusters as
$\mathbf{\underaccent{\tilde}x} =(x_1, x_2,..., x_{10})$, $\mathbf{\underaccent{\tilde}y} =(y_1, y_2,..., y_{10})$ and $\mathbf{\underaccent{\tilde}z} =(z_1, z_2,..., z_{10})$. Here $x_i\sim Bin\left(1,1-p\right)$ for $i = 1,2,3,4,5$, $x_i\sim Bin\left(1,p\right)$ for $i = 6,7,8,9,10$, $y_i\sim Bin\left(1,p\right)$ for $i = 1,2,3,4,5$, $y_i\sim Bin\left(1,1-p\right)$ for $i = 6,7,8,9,10$, and $z_i\sim Bin\left(1,p\right)$ for $i = 1,2, ...,10$. 
Each cluster contains 100,000 observations and these 300,000 observations are then placed together in a dataset.

Since these observations are binary, there are at most $2^{10}=1024$ possible unique observations. Each of these observations will represent a data nugget, and the weight of the data nugget is the number of times this data nugget center appears in the dataset. Using the K-means clustering for weighted observations method given by \textbf{Algorithm \ref{alg3}}, we assign each data nugget to one of three clusters. In addition, we also fit the K means cluster model for the raw data as a benchmark here. Then, we find the proportion of correct cluster assignments for every possible permutation of cluster assignments and choose the cluster configuration that produces the best result for each method.

Repeating this process for 100 iterations, we found the mean proportion of data points correctly reassigned to their proper cluster for each method for different $p$.  The simulation results are given in Table \ref{tab:p8}. It is clear to see that, compared with the benchmark result using Kmeans for the original data directly, the K-means for weighted observations algorithm performs exactly the same here. 

\begin{table}[h!]  

\footnotesize
\centering
 \begin{tabular}{|c|c|c|c|c|c|c|c|} \hline
 $p$ & 0.80 & 0.82 & 0.84 & 0.86 & 0.88 & 0.90  \\ \hline
  K-means for weighted observations & 0.9185&0.9388&0.9558&0.9661&0.9803&0.9883\\ \hline
   K-means for raw data & 0.9185&0.9388&0.9558&0.9661&0.9803&0.9883\\ \hline 
 \end{tabular} 
\label{tab:p8} 
\caption{Correct Cluster Classification Simulation Results}
\end{table}

\textbf{Example 2}\\ 
A simulated large multivariate data set with known cluster structures was employed and compared with other state-of-the-art models. The simulated dataset is a large 6-dim data with four unequal-sized clusters and its overall size is 1052000. It was generated by firstly assigning four cluster centers, and then obtaining data for each cluster by sampling from a multivariate Gaussian distribution with mean as the cluster center and a specified covariance matrix. The covariance matrices were the identity times a constant variance. The cluster centers, sizes, and variances are listed in Table \ref{table: accuracy}. The challenge in this dataset is the small $3^{\text {rd }}$ and $4^{\text {th }}$ clusters, which would be hard to find within the large dataset.

\begin{table}[h!]
\footnotesize
\centering
\label{table: accuracy} 
\begin{tabular}{|l|l|l|l|l|}
\hline & Cluster 1 & Cluster 2 & Cluster 3 & Cluster 4 \\
\hline cluster size & 500000 & 500000 & 50000 & 2000 \\
\hline cluster center & $(1,0,0,0,1,1)$ & $(0,1,0,1,1,0)$ & $(1,1,0,0,1,0)$ & $(0,0,1,1,0,1)$ \\
\hline 
Variance on each
dim & 0.25 & 0.25 & 0.25 & 0.25 \\
\hline
\end{tabular} 
\caption{Cluster information for the simulated data}
\end{table}
For this large dataset, applying clustering algorithms on the whole dataset is not computationally possible, but the data nuggets created and refined by the proposed method could be used to reduce the data size while preserving the data structure. The clustering results may be obtained by applying the weighted K-means method to the data nuggets. For comparison, the random sampling method with $\mathrm{K}$-means and the BIRCH clustering method were also employed on this dataset. $\mathrm{BIRCH}$ (balanced iterative reducing and clustering using hierarchies) is a commonly used algorithm to perform clustering over particularly large datasets \citep{saeed2020big, shirkhorshidi2014big, zhang1996birch}. It first builds a clustering feature (CF) tree out of the data points, and then applies hierarchical clustering to cluster all leaf entries \citep{zhang1997birch}. To account for the variability of each method, the data nugget creation and refinement with weighted K-means, and the random sampling with K-means, were performed ten times. During each time, around 3200 data nuggets were created and refined, and for comparison, random samples with 3200 points were generated from the whole data. The BIRCH method was also employed ten times with the threshold parameter randomly sampled from a uniform distribution Unif $(0.7,1.2)$. The threshold parameter affects the merge of subclusters, and lower values promote splitting during building the $\mathrm{CF}$ tree. The data nuggets and the random sampling methods were implemented using $\underline{\mathrm{R}}$ but the $\mathrm{BIRCH}$ method would not run in $\mathrm{R}$ for such a large data set, instead the python implementation of $\mathrm{BIRCH}$ was able to run the ten times.

To demonstrate the performance of clustering, the clustering results were firstly aligned with the true clusters, and then the accuracies were calculated for each cluster. The boxplots of the accuracies from ten repetitions by different methods are presented in Figure \ref{fig:compare} for each cluster separately.

\begin{figure}[h!]

\centerline{\includegraphics[width = 340pt]{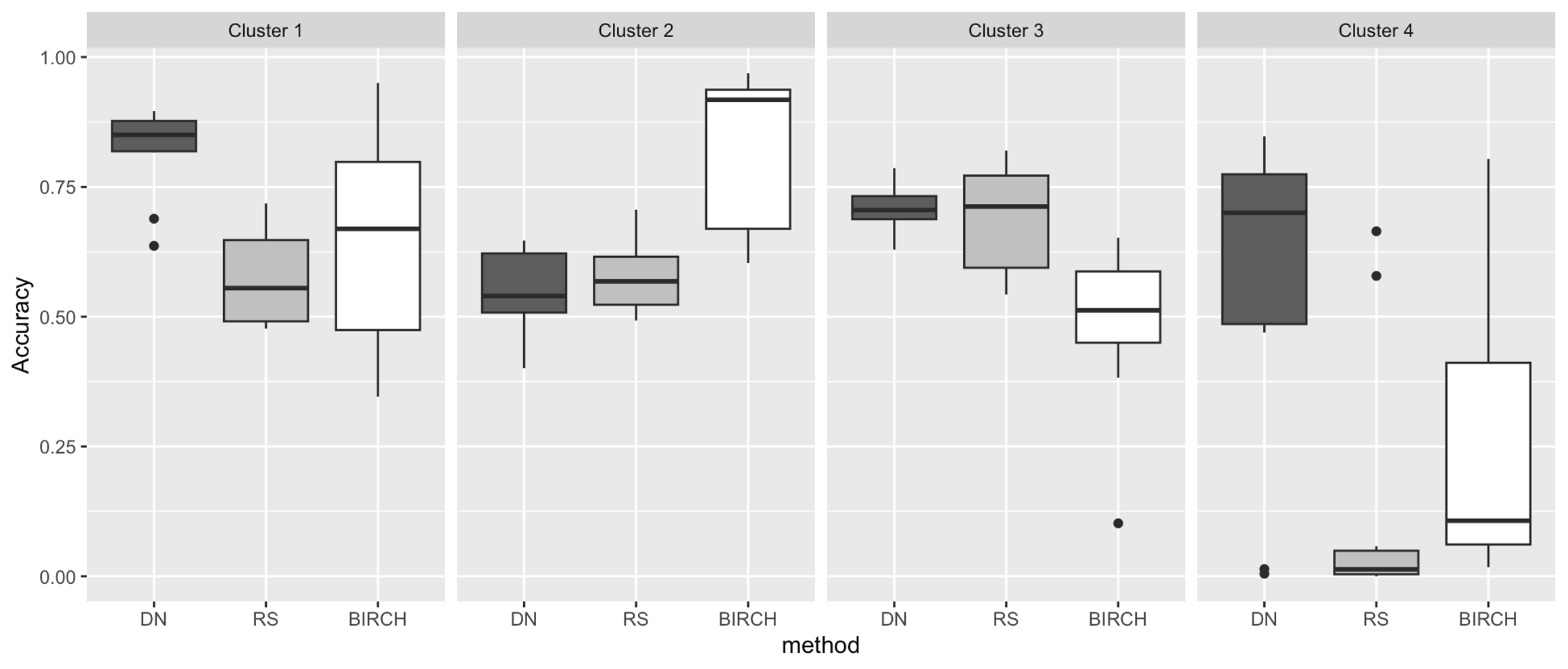}}

\caption{Boxplots of accuracies for each cluster by data nugget method (DN), random sampling (RS), and BIRCH algorithm.}

\label{fig:compare}

\end{figure}
Based on this example, distinct methodological performances were observed across various clusters. The data nugget method exhibited the highest mean accuracy and the smallest variance for the first primary cluster, while BIRCH performed well in terms of mean accuracy for the second main cluster but showed a larger variance indicating less consistency. When dealing with the third, smaller cluster, the data nugget method achieved both a high mean accuracy and the lowest variance, highlighting its effectiveness in handling smaller clusters. In the most challenging scenario, cluster 4, the smallest in size, the data nugget method significantly outperformed all others, with only two outliers failing to identify the cluster. Random sampling struggled to locate this smallest cluster within the large dataset, with only two samples containing information about the cluster that achieved high accuracies. The results from this simulated dataset revealed that random sampling could identify the two main clusters but did so with a higher variance when it came to the third, smaller cluster. Moreover, it faced considerable difficulty in detecting the smallest cluster within the extensive dataset. The BIRCH algorithm exhibited greater variability in its results due to its threshold parameter, and it did not perform well when dealing with small clusters. In contrast, the data nugget method consistently found all clusters with high accuracies and low variances, especially excelling in exploring the two smaller clusters within the vast multivariate dataset.

In summary, random sampling is computationally very efficient, but it may sacrifice some information about the data structure during the sampling process, particularly smaller structures within large datasets may not be detected. Similarly, the performance of the BIRCH algorithm depended heavily on the threshold parameter chosen. However, the created and refined data nuggets could effectively preserve hidden structures within big data that are commonly overlooked. These data nuggets could be treated as weighted data points and applied to existing methods to enhance the exploration of the structure of large datasets.

\subsection{Data Nuggets vs. Support Points}

In Section 1 we mentioned another method of producing representative datasets called support points given by Mak and Joseph. The goals of the data nuggets and the support points are the same on the surface: to create a small dataset that represents the large dataset it comes from. That being said, the resulting representative datasets may differ greatly in terms of producing the correct quantiles corresponding to the highest percentiles of the probability distributions they are meant to represent.

This is by design in the case of support points, since they are defined as a set of M points in the dataset which has the best goodness of fit to the underlying distribution governing the dataset in terms of energy distance as defined in \citep{Sze2013}. This definition forces support points to be chosen as observations that exist near the center of the data, ultimately forsaking observations that exist at the extreme edges of the data.

Data nuggets on the other hand are designed to avoid this problem. Since the algorithm that creates data nuggets chooses observations to delete based on how close they are to other observations, observations on the edges of the data are safer from elimination and are usually guaranteed to remain as an initial data nugget center. The fact that this observation is at the edge of the data is not lost since it will be reflected in a low-weight parameter for the data nugget.

\begin{figure}[h!]

\centerline{\includegraphics[width=340pt]{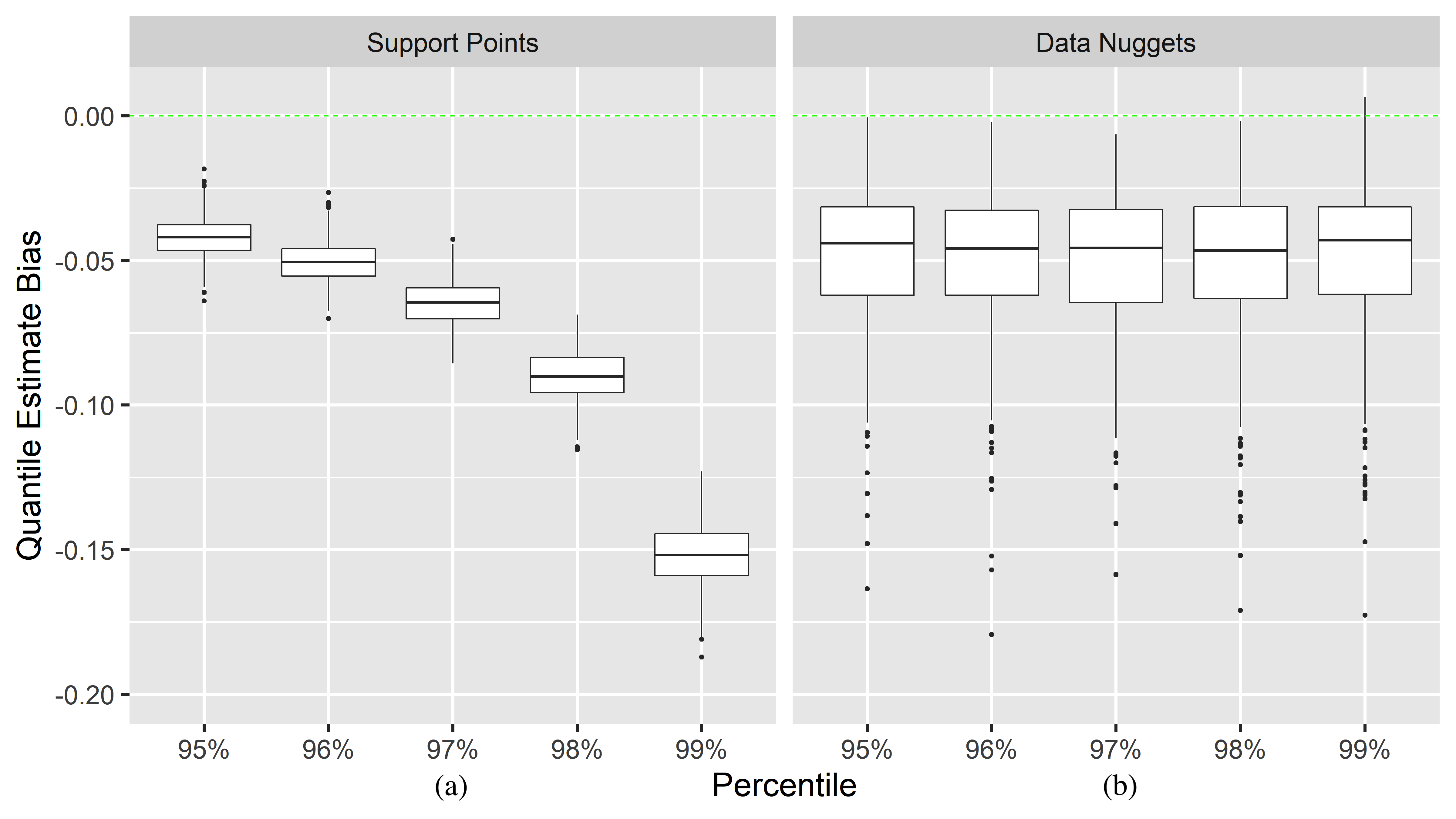}}

\caption{Quantile Bias Simulation Results}

\label{fig:QBias}

\end{figure}

We now produce the results of a simulation that examines this difference in a 1-dimensional setting. The simulation was conducted by randomly sampling 100,000 observations from a standard normal distribution. Let this random sample of observations be $\mathbf{\hat{\underaccent{\tilde}z}}$. 100 support points and 100 data nuggets are then created from $\mathbf{\hat{\underaccent{\tilde}z}}$. The support points were generated using the R package ``support'' created by Mak \citep{SoloMak2019}. The data nuggets were generated with \textbf{Algorithm \ref{alg1}} by choosing $\mathbf{X}=\mathbf{\hat{\underaccent{\tilde}z}}$, data nugget centers chosen to be the mean, $R = 5000$, $C = .1$, $M_{init}=1,000$, $M=100$, and $D$ to be the Euclidean distance metric. The data nuggets are then ordered by their centers, from least to greatest.

We then compute the quantiles corresponding to the ${95}^{th}$, ${96}^{th}$, ${97}^{th}$, ${98}^{th}$ and ${99}^{th}$ percentiles for each representative dataset. The quantiles for the support points are calculated in the typical fashion; however, calculating the quantiles for the data nuggets requires a more thoughtful process. 

First, a linear regression model is fit with the cumulative sums of the data nugget weights divided by 100,000 as the predictor variable and the data nugget centers as the response variable. Then, .95, .96, .97, .98, and .99 are plugged into the resulting regression equation to produce the quantiles corresponding to those percentiles for the data nuggets. 

Finally, the true quantiles for a standard normal distribution are subtracted from the quantiles calculated for each method to calculate the bias of each quantile for each method. This simulation was repeated for 500 sets of $\mathbf{\hat{\underaccent{\tilde}z}}$ and Figure \ref{fig:QBias} shows the results. Figure \ref{fig:QBias} and all figures that follow were produced using the R package ``ggplot2'' \citep{ggplot2016}. The box plots in Figure \ref{fig:QBias} represent the distribution of quantile estimate bias for each corresponding percentile for each method.

 It is clear to see that support points perform poorly in terms of median bias compared to data nuggets for calculating the quantiles at the upper tail of the normal distribution, even though there is a little bit of mean bias due to the skewness. Also note that since the bias for the quantiles given by the data nuggets is consistent across the percentiles, there may be a simple correction constant that can be applied to each quantile to eliminate the bias entirely.

\subsection{Effectiveness When $P$ is Large}

When the number of variables is large, questions can reasonably be raised about whether or not data nuggets remain effective at capturing the structure of the entire dataset. A simulation was created to examine this question.

First, the $600,000 \times 3$ data matrix $\mathbf{X_1}$ was formed by combining the observations yielded from randomly sampling 200,000 vectors from $N_3((0,0,10)^{\prime}, \mathbf{\Sigma_{X}})$, 200,000 vectors from $N_3 \Big( \Big( 0,\frac{6}{\sqrt{2}},\frac{6}{\sqrt{2}} \Big)^{\prime}, \mathbf{\Sigma_{X}} \Big)$, and 200,000 vectors from $N_3 \Big( \Big( \frac{10}{\sqrt{3}},\frac{10}{\sqrt{3}},\frac{10}{\sqrt{3}}\Big)^{\prime}, \mathbf{\Sigma_{X}} \Big)$, where: 

\[ 
\mathbf{\Sigma_{X}} 
=
\begin{bmatrix}
4 & 0 & 0 \\
0 & 2.25 & 0 \\
0 & 0 & 1
\end{bmatrix}
\]

Second, the $600,000 \times 3$ data matrix $\mathbf{X}_2$ was formed by randomly sampling 600,000 vectors from $N_{197}(\mathbf{0}^{\prime}, \mathbb{I}_{197})$, where $\mathbf{0}$ is the vector containing only 0's. Third, the data matrix $\mathbf{X}_3$ was formed by horizontally concatenating $\mathbf{X}_1$ and $\mathbf{X}_2$ so that $\mathbf{X}_3 = [\mathbf{X}_1 \mathbf{X}_2]$. Fourth, a $200 \times 200$ random rotation $\mathbf{T}$ was created by randomly sampling 200 vectors from $N_{200}(\mathbf{0}^{\prime}, \mathbb{I}_{200})$. Finally, the data matrix $\mathbf{X}$ was formed by applying the random rotation $\mathbf{T}$ to $\mathbf{X}_3$ so that $\mathbf{X_T} = \mathbf{T}\mathbf{X}_3$. 

2,000 data nuggets were then generated with \textbf{Algorithm \ref{alg1}} by choosing $\mathbf{X}=\mathbf{X_T}$, data nugget centers chosen to be random, $R = 5,000$, $C = .1$, $M_{init}=10,000$, $M=2,000$, and $D$ to be the Euclidean distance metric. To compare the structure of these two datasets, principal components were generated for $\mathbf{X_T}$, and principal components for weighted observations were generated for the 2,000 data nuggets. All principal components were found using the ``Wpca'' function from the developed R package ``WCluster''. Note that the weights entered for the ``Wpca'' function when creating the principal components for the entire dataset were all equal to 1 while the principal components for data nuggets were weighted according to the weights of the data nuggets. Also, note that some principal components have been multiplied by -1 to make the similarity in the structures of the data more apparent.



\begin{figure}[h!]

\centerline{\includegraphics[width=340pt]{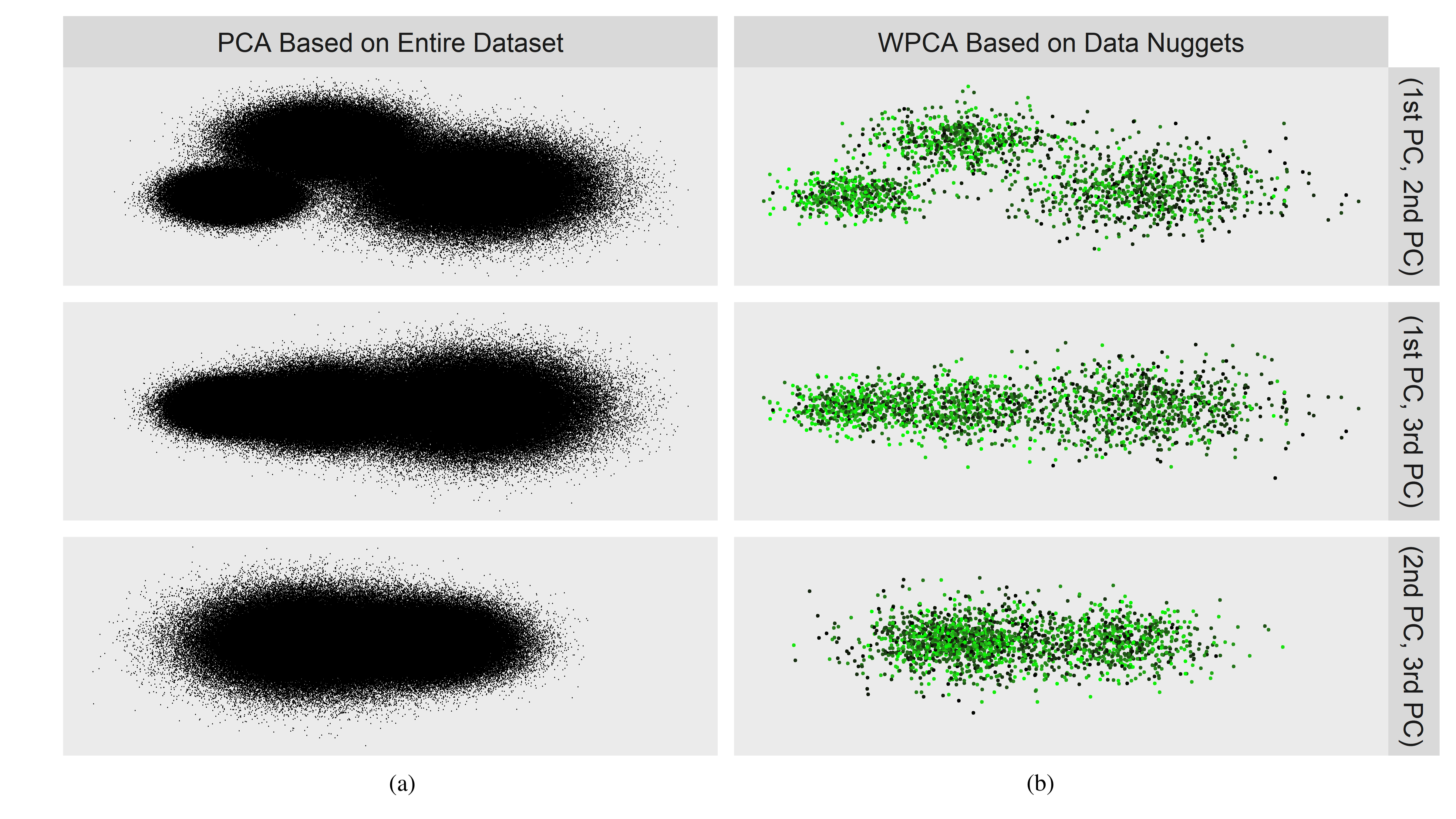}}

\caption{PCA Plots of $\mathbf{X_T}$ vs WPCA Plots of Data Nuggets}

\label{fig:largeP}

\end{figure}

In Figure \ref{fig:largeP}, pairwise combinations of the first, second, and third principal components of the entirety of $\mathbf{X_T}$ (a) are shown beside the same pairwise combinations of the first, second, and third weighted principal components of the 2,000 data nuggets (b). The intensity of each data nugget corresponds to the weight of the data nugget using weight median as the threshold: lighter intensity indicates a large weight while darker intensity indicates a low weight.  Clearly, the data nuggets reproduce the structure of the first three principal components of $\mathbf{X_T}$ in its entirety.



\section{Application}

Lastly, we applied this algorithm to the analysis of a high-dimensional ($1,048,575 \times 9$) flow cytometry dataset generated following protein immunization in mice.

We begin by using \textbf{Algorithm \ref{alg1}}, choosing data nugget centers chosen to be the mean, $R = 5,000$, $C = .01$, $M_{init} = 10,000$, $M = 2,000$, and $D$ to be the Euclidean distance metric to create 2,000 data nuggets. We then refined the data nuggets with \textbf{Algorithm \ref{alg2}}, choosing $\nu = .25$ and $N_{min}=2$ which resulted in 3,135 data nuggets.  

\begin{figure}[h!]

\centerline{\includegraphics[width=340pt]{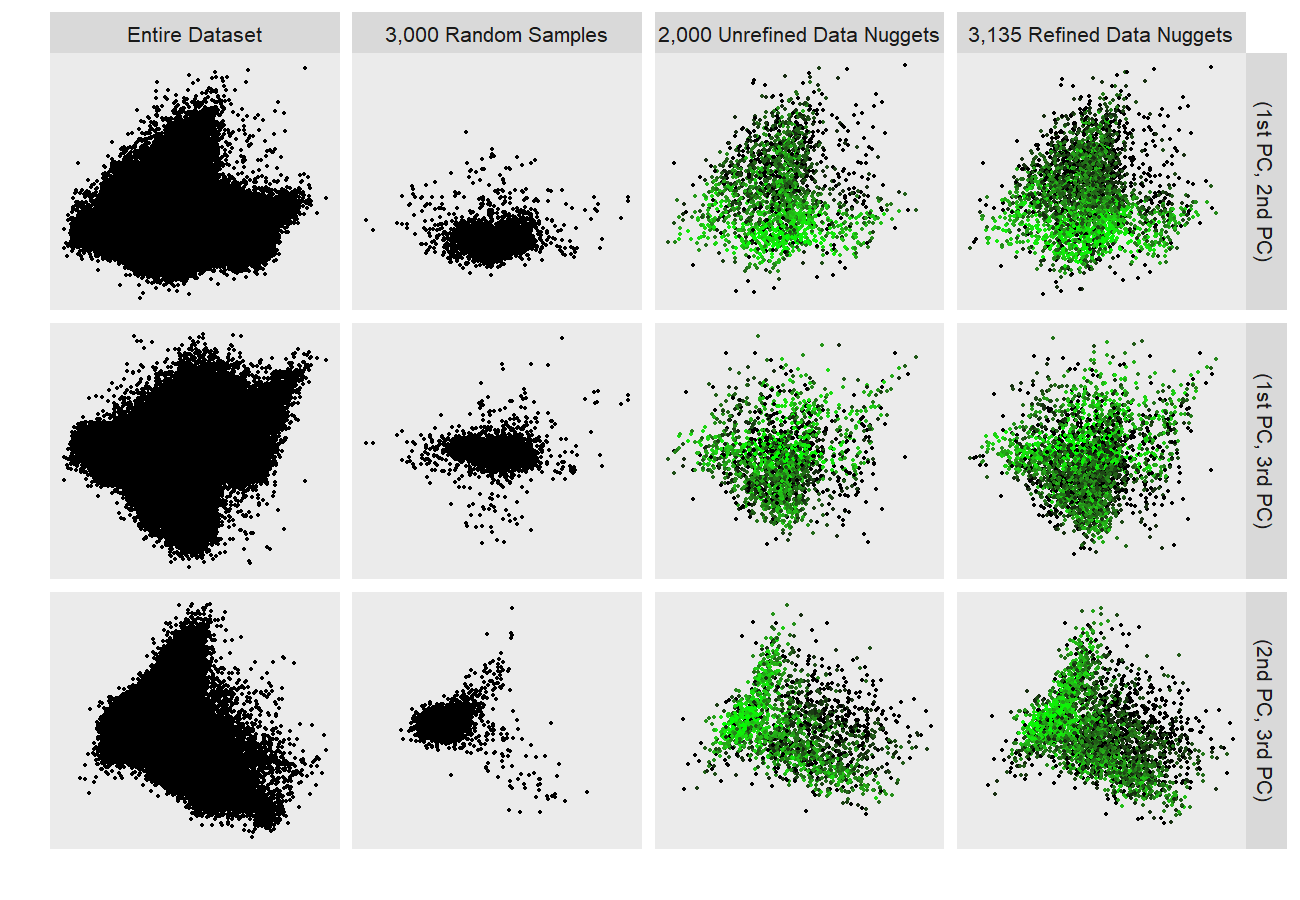}}

\caption{PCA Plots of Entire Dataset vs WPCA Plots of Data Nuggets}

\label{fig:realdata1}

\end{figure}


The first three pairwise combinations for principal components of the entire dataset (a) and 3000 random data samples (b), for principal components for weighted observations of the initial 2,000 data nuggets (c) and the refined 3,135 data nuggets (d) are given in Figure \ref{fig:realdata1} to compare of the resulting data structures. Once again, note that the weights entered for the ``Wpca'' function when creating the principal components for the entire dataset were all equal to 1 while the principal components for the data nuggets were weighted according to the weights of the data nuggets. Also note that some of the component scores for some principal components have been multiplied by -1, swapped, and/or shifted to make the similarity in the structures of the data more apparent. Once again, the intensities of each data nugget corresponds to the weight of the data nugget. Lighter intensity indicates weight larger than mean weight while darker intensity indicates weight smaller than mean weight. 

Figure \ref{fig:realdata1} also shows that the structure of the data regarding the first three principal components is moderately recovered with the original 2,000 data nuggets and strongly recovered with the 3,135 refined data nuggets, but not recovered well by the 3000 random samples from the data. Some information about the data structure is missed by random sampling but the data nuggets maintain it. Recall that the original dataset contains over 1 million observations, so the fact that less than 1\% of these observations can be chosen and still produce a relatively strong representation of the structure of the data is noteworthy.

\begin{figure}[h!]

\centerline{\includegraphics[width=\textwidth]{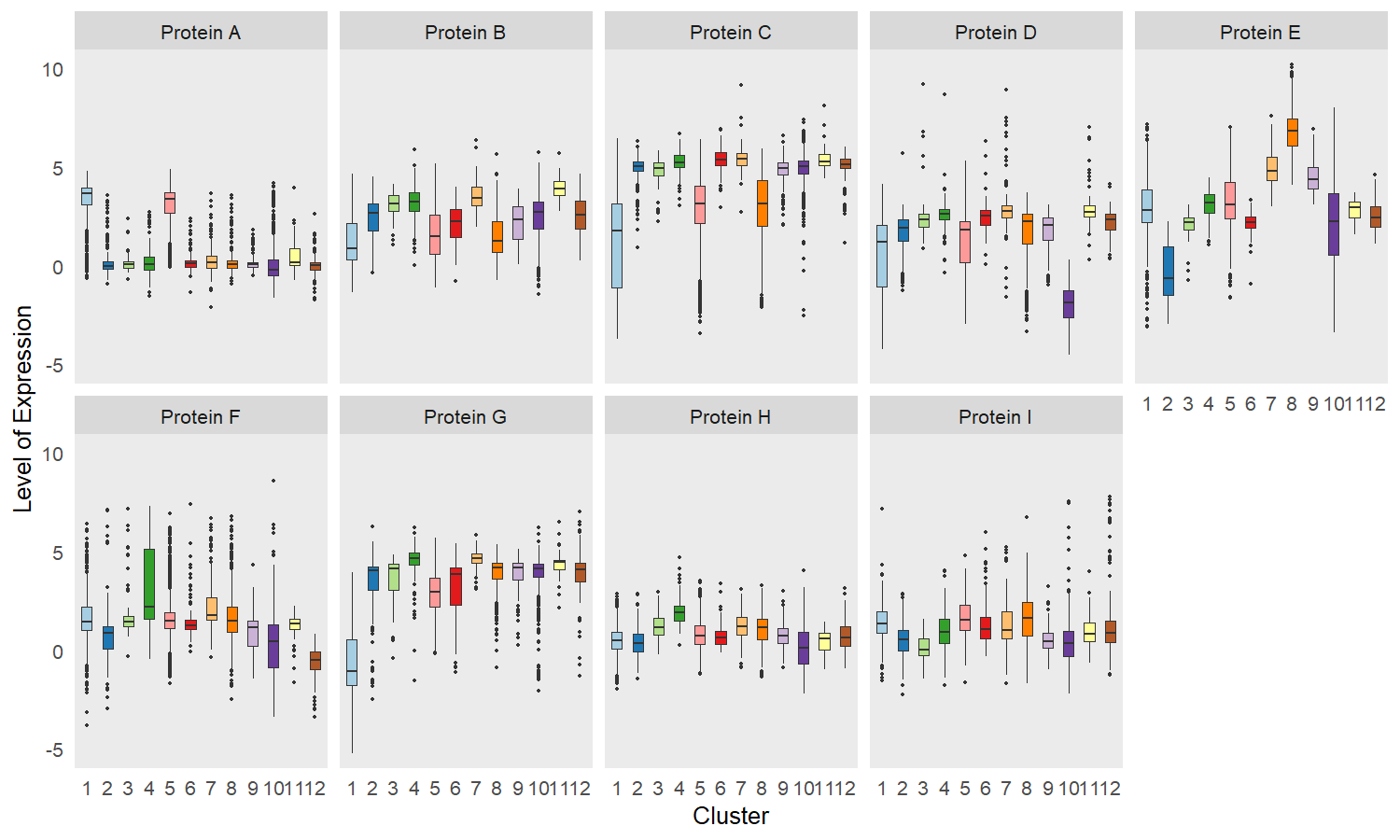}}

\caption{Levels of Expression for Each Protein and Cluster Combination}

\label{fig:realdata3}

\end{figure}

Next, we configured the refined data nuggets into clusters using \textbf{Algorithm \ref{alg3}} to perform weighted K-means clustering. Ten initial centers were used and the cluster configuration with the least weighted within-cluster sum of squares (WWCSS) was chosen. Different numbers of clusters were chosen from 5 to 15, and compared by WWCSS for each number of clusters. The best number of clusters was chosen as 12 using the second difference method which is described in Section 3.1. The protein level profiles in different clusters enable scientists to identify the cell type and function of the cells in each cluster. Ultimately, these clusters will result in counts for each cell type and function that will be used for further statistical analysis and biological interpretation. Finally, we created box plots for each cluster which summarize the level of expression within the cluster for each protein  (measured via each component of the data nugget center) to search for whether any clusters show any visually significant levels of expression of any proteins. These box plots are given in Figure \ref{fig:realdata3}, and could then be examined by the scientists to see if any meaningful clusters have appeared.

\section{Conclusion}

We have detailed a method for reducing ``Big Data" using data nuggets. First, we offer a K-means algorithm for weighted observations to cluster these data nuggets and provide simulation results that show that this algorithm outperforms the K-means clustering algorithm for data nuggets yielded from binary data. Then, we displayed the distinction between data nuggets and support points in the context of quantile bias at the upper tail of a normal distribution using a simulation, showing that there is a greater level of bias when these quantiles are calculated with support points. Further, we used a simulation to demonstrate that even for datasets with large $P$, data nuggets are capable of capturing the structure of the dataset. This framework clearly illustrates two major advantages of data nuggets. First, the standard method implemented in R like k-means clustering and hierarchical cluster ($hclust$ function or $pam$ function) cannot be run for the data is over 1 million observations, however, we can easily obtain the result through k means for weighted observations and the result here is very comparable to k mean on the raw dataset, which is shown by one of our binary data examples. Second, data nuggets can maintain the data structure well, and we compare with random samples for face example, support points for simulated normal distribution, and PCA on random samples for flow cytometry data, all of them show that you can’t capture the structure well if using these three method compared to data nugget.

The R packages ``datanugget''  and ``WCluster'' have been developed to execute the methods described in this paper. They include functions for generating, refining, and clustering data nuggets using K-means clustering for weighted observations. They have been published and are available on CRAN now.

Future work could be done to show how well the data nuggets work when other mainstream statistical techniques are applied. We have already shown how well data nuggets can work when unsupervised methods such as principal components and clustering are applied. Another unsupervised method of interest that could be applied is projection pursuit \citep{Fri1974}. The efficacy of data nuggets could also be observed in the context of supervised methods such as logistic regression and linear regression.  

In the case of logistic regression, the response for each data nugget would be the number of ``successful" and ``unsuccessful" observations contained in the data nugget. In the case of linear regression, the response for each data nugget would be the mean of the responses of the observations contained in the data nugget, and weighted least squares regression could be applied. The weight of each data nugget (potentially combined with the variance of the response variable for each data nugget) would be used as the weight in the regression model.

An important area of improvement for this method would be to find the optimal number of data nuggets for any given sample size. Simulations involving large classified continuous datasets could also be created to determine how much better K-means clustering for weighted observations performs compared to K-means clustering of data nuggets in a continuous setting. 

Another area of interest is showing that the results of the simulation in Section 3.2 hold for higher dimensions. Work could also be done to provide a correction for the constant bias in estimating the quantiles with data nuggets. Research into asymptotic results regarding how well the probability distribution can be returned through estimation of the mean and covariance of data nuggets generated from a random sample of this probability distribution as the number of data nuggets increases to infinity would be useful as well. 

\newpage

\bigskip

\begin{center}
{\large\bf SUPPLEMENTARY MATERIAL}
\end{center}

\begin{description}

\item[R packages for algorithms:] R packages ``datanugget'' and  ``WCluster'' containing functions that perform algorithms described in this paper. (tar.gz file)

\item[Flow Cytometry Dataset used in Section 4:] Flow cytometry dataset that has been masked, permuted, and had random noise added to it. (.RData file)

\item[R Code for output:] Code that produces all figures and tables found in this paper. (.R file)

\item[Appendix on computational cost:]
The computational cost of this algorithm contains three parts$O \Big( P* \Big[ G(\frac{R(R-1)}{2})(\psi_1) + \frac{(M_{init})(M_{init}-1)}{2}(\psi_2) + (N+1)M \Big] \Big)$,  in terms of arithmetic complexity can be defined by $O \Big( P* \Big[ G(\frac{R(R-1)}{2})(\psi_1) + \frac{(M_{init})(M_{init}-1)}{2}(\psi_2) + (N+1)M \Big] \Big)$ where $\psi_1$ and $\psi_2$ are the number of iterations it takes to reduce $R$ to $\lceil \frac{M_{init}}{G} \rceil$ and $M_{init}$ to $M$, respectively. As such, $2 \leq \psi_1 \leq R - \frac{M_{init}}{G} + 1$ and $2 \leq \psi_2 \leq M_{init} - M + 1$. Where $\psi_1$ and $\psi_2$ fall within their respective ranges depends on the choice of $C$. Specifically:
$$\lim_{C \rightarrow 0} \psi_1 = R - \frac{M_{init}}{G} + 1, \hspace{0.1in} \lim_{C \rightarrow 0} \psi_2 =  M_{init} - M + 1, \hspace{0.1in} \lim_{C \rightarrow 1} \psi_1 = \lim_{C \rightarrow 1} \psi_2 =  2$$
\end{description}

\newpage

\bibliographystyle{asa}

\bibliography{DataNuggetsBibliography}

\end{document}